\documentclass[12pt,reqno]{amsart}

\newcommand\version{June 10, 2018}

\usepackage{amsmath,amsfonts,amsthm,amssymb,amsxtra}
\usepackage{bbm} % to get \1
\usepackage{tikz}

\newtheorem{theorem}{Theorem}%[section]
\newtheorem{proposition}[theorem]{Proposition}
\newtheorem{lemma}[theorem]{Lemma}
\newtheorem{corollary}[theorem]{Corollary}
\newtheorem{conjecture}[theorem]{Conjecture}

\theoremstyle{definition}

\theoremstyle{remark}

\newtheorem{remark}[theorem]{Remark}

\newcommand{\C}{\mathbb{C}}

\renewcommand{\epsilon}{\varepsilon}

\renewcommand{\phi}{\varphi}
\newcommand{\R}{\mathbb{R}}

\newcommand{\cH}{{\mathcal H}}

\DeclareMathOperator{\im}{Im}

\DeclareMathOperator{\spec}{spec}
\DeclareMathOperator{\supp}{supp}

\DeclareMathOperator{\tr}{Tr}

\begin{document}

\title[Inequalities for quantum divergences --- \version]{Inequalities for quantum divergences\\ and the Audenaert--Datta conjecture}

\author{Eric A. Carlen}
\address[Eric A. Carlen]{Department of Mathematics, Hill Center,
Rutgers University, 110 Frelinghuysen Road, Piscataway, NJ 08854-8019, USA}
\email{carlen@math.rutgers.edu}

\author{Rupert L. Frank}
\address[Rupert L. Frank]{Mathematisches Institut, Ludwig-Maximilans Universit\"at M\"unchen, Theresienstr. 39, 80333 M\"unchen, Germany, and Mathematics 253-37, Caltech, Pasa\-de\-na, CA 91125, USA}
\email{rlfrank@caltech.edu}

\author{Elliott H. Lieb}
\address[Elliott H. Lieb]{Departments of Mathematics and Physics, Princeton University, Washington Road, Princeton, NJ 08544, USA}
\email{lieb@princeton.edu}

\begin{abstract}
Given two density matrices $\rho$ and $\sigma$, there are a number of different expressions that reduce to the $\alpha$-R\'enyi relative entropy of $\rho$ with respect to $\sigma$ in the classical case; i.e., when $\rho$ and $\sigma$ commute. Only those expressions for which the Data Processing Inequality (DPI) is valid are of potential interest as quantum divergences in quantum information theory. 
Audenaert and Datta have made a conjecture on the validity of the DPI for an interesting family of quantum generalizations of the 
$\alpha$-R\'enyi relative entropies, the $\alpha-z$--R\'enyi relative entropies. They and others have contributed to the partial solution of this conjecture. We review the problem, its context, and the methods that have been used to obtain the results that are known at present, presenting a unified treatment of developments that have unfolded  in a number of different papers.
\end{abstract}

\thanks{\copyright\, 2018 by the authors. This paper may be
reproduced, in its entirety, for non-commercial purposes.\\
U.S.~National Science Foundation grants DMS-1501007 (E.A.C.), DMS-1363432 (R.L.F.), PHY-1265118 (E.H.L.) are acknowledged.}

\maketitle

\section{The Audenaert--Datta conjecture and known results}

The  sender of a message over a classical noisy channel communication encodes it into a sequence of characters from an alphabet -- possibly just $\{0,1\}$ -- 
that are physically represented by levels (amplitude or frequency) in a transmitted signal. Because of noise, the received levels will be random 
variables with continuous and overlapping distributions. The alphabet and encoding are known to the recipient, who is, however,  faced with the problem of deciding whether each incoming  random signal level represents, say,  $0$ or $1$. 

Shannon's theory tells the sender and the recipient  just how much redundancy they must employ for the recipient to correctly extract the message from the noisy signal that is received with an arbitrarily small probability of error.  
For example, suppose the communication channel they use is such that sending $0$ through the channel results in a random variable with a density $\rho$ centered on $0$, while sending $1$ through the channel results in a random variable with density $\sigma$ centered on $1$. Suppose the noise is such that
the different random variables produced at each step of the transmission are independent. Let $X_n$ denote the $n$-th signal received. 
If  the densities $\rho$ and $\sigma$ overlap (and they will if the noise is Gaussian) the receiver cannot tell for sure what was sent on a single observation. But if the sender repeats the transmission of the {\em same} signal $m$ times, and the receiver knows it is the same signal being sent $m$ times, and makes optimal use of the observed signals, then the probability of reading the wrong bit will go to zero exponentially fast as $m$ increases. 

Thus, the following problem is fundamental to classical communication theory:  Given two probability distributions
$\rho$ and $\sigma$ on $\R$, and a sequence of random variables $X_n$ drawn from one of these two distributions, decide on the basis of the observations $\{X_1,\dots,X_N\}$ whether $\rho$ or $\sigma$ is the distribution from which the random sequence is being drawn.  That is, one has to come up with a set $A_N\subset \R^N$ so that if 
$(X_1(\omega),\dots,X_N(\omega))\in A_N$, then one accepts that $\rho$ is the governing distribution, while otherwise, one accepts that $\sigma$ is the governing distribution.

There are two kinds of errors that one can make: accepting $\sigma$ when $\rho$ is the governing distribution, and accepting $\rho$ when $\sigma$ is the governing distribution.  Of course, if one takes $A_N$ to be all of $\R^N$, then one never makes the  first kind of error, but one will make the second kind of error whenever $\sigma$ is the governing distribution. Therefore, fix some small $\epsilon>0$, and require of $A_N$ that 
$\int_{A_N}\rho^{\otimes N} > 1-\epsilon$.
Then among all such choices for $A_N$, choose one that (nearly) minimizes $\log \int_{A_N}\sigma^{\otimes N}$. That is, define
\begin{equation}\label{lardev}
\beta_{\epsilon,N}(\rho,\sigma) = \inf\left\{ \log \int_{A_N}\sigma^{\otimes N}\ : \  A_N\subset \R^N \ {\rm is\ such\ that}\ 
\int_{A_N}\rho^{\otimes N} \geq 1-\epsilon\ \right\}.
\end{equation}
Then one has \cite{KL51,K67}
\begin{equation}\label{uplow}
\limsup_{N\to\infty} \frac1N \beta_{N,\epsilon}((\rho,\sigma) \leq -D(\rho||\sigma)  \quad{\rm and}\quad
\liminf_{N\to\infty} \frac1N \beta_{N,\epsilon}((\rho,\sigma) \geq -\frac{1}{1-\epsilon}D(\rho||\sigma)\, ,
\end{equation} 
where $D(\rho||\sigma)$ is the {\em relative entropy} or {\em Kullbach--Liebler divergence} of $\rho$ with respect to $\sigma$:
\begin{equation}\label{kull}
D(\rho||\sigma) = \int_\R \rho(x)(\log \rho(x) - \log \sigma(x)){\rm d}x\, .
\end{equation}
If one chooses $N$ so that $e^{-ND(\rho||\sigma)} <  \epsilon$, then one can expect to have made both types of errors small, of order $\epsilon$.

It is an easy consequence of Jensen's inequality that $D(\rho||\sigma)\geq 0$ with equality if and only if $\rho =\sigma$.
However, it is not in general true that $D(\rho||\sigma) = D(\sigma||\rho)$, and so the relative entropy is not a metric on the space of probability distributions.   The asymmetry directly reflects the asymmetry in the question that is answered in terms of  the relative entropy, namely: {\em If one chooses the acceptance rule $A_N$ so that the probability of correctly accepting $\rho$ is at least $1-\epsilon$, how small, as a function of $N$,  can one make the probability of incorrectly accepting $\rho$ when  $\sigma$ is sent?}  In this simple setting of independent random variables, the fact that the relative entropy arises as the answer to this question is a consequence of Cram\'er's theorem on large deviations, and it gives the relative entropy an {\em operational meaning}. 

A {\em (classical) divergence}  is a  function on pairs of probability  
densities taking values in $[0,\infty]$ that is somehow connected with how  ``distinguishable'' the two densities are, and as above, such functions need not be symmetric.  A mathematical definition of classical divergences  
was given by R\'enyi \cite{Ren} who introduced the R\'enyi relative entropies as a family of such divergences. For 
$\alpha\in (0,1)$, the $\alpha$-R\'enyi entropy is defined as
\begin{equation}\label{aren}
D_\alpha(\rho||\sigma) = \frac{1}{\alpha-1} \log \left(\int_\R\rho^{\alpha}(x) \sigma^{1-\alpha}(x){\rm d}x\right)\ .
\end{equation}
Later, Csizl\'ar \cite{Csi} gave an operational meaning to the R\'enyi relative entropies, but going into this would be too large a digression. Suffice it to say that a great many bounds on error probabilities have been given in terms of R\'enyi entropies. Our focus here is on the quantum aspects of this problem, and especially, quantum analogs of a certain  monotonicity property  that classical divergences must have. The basic monotonicity property can be easily  explained at an intuitive level for both the classical and quantum cases.

Any divergence is supposed to give, or at least bound,  the ``best asymptotic rate of distinguishability'' between $\rho$ and $\sigma$ in some operational context. Let $P(x,y)$
be a non-negative kernel with $\int P(x,y){\rm d}x = 1$ for all $y$. Define $P\rho(x) := \int_\R P(x,y)\rho(y){\rm d}y$
and $P\sigma(x) := \int_\R P(x,y)\sigma(y){\rm d}y$. Any classical divergence $D(\rho||\sigma)$ must satisfy
\begin{equation}\label{clDPI}
D(P\rho ||P\sigma) \leq D(\rho ||\sigma)
\end{equation}
since otherwise, one could apply the operation $P$ to $\rho$ and $\sigma$, and use whatever procedure one is using to distinguish $\rho$ and from $\sigma$, and get a better result. The inequality \eqref{clDPI}
which reflects the fact that applying a further random corruption of the signal can only make it harder to discern what is being sent, is known as the {\em (classical) Data Processing Inequality}. For the classical divergences discussed above, this is true on account of Jensen's inequality and the joint concavity of the integrands in $\rho$ and $\sigma$.  There are other important properties that divergences should have; e.g., certain additivity properties over products, but we focus instead on \eqref{clDPI} and its quantum analog. 

 In the quantum setting, probability densities are replaced by density matrices; that is, by non-negative trace class matrices $\rho$ with unit trace, and integrals are replaced by traces.  One natural quantum analog of the relative entropy $D(\rho||\sigma)$, known as the {\em Umegaki relative entropy} \cite{U62}, is defined to be
\begin{equation}\label{Um}
D(\rho ||\sigma) = \tr[\rho (\log \rho - \log \sigma)]\ ,
\end{equation}
which is closely related to the {\em von Neumann entropy} of a density matrix $\rho$:
\begin{equation}\label{vnent}
S(\rho) = -\tr[\rho \log \rho]\ .
\end{equation}
The fact that $D(\rho ||\sigma)\geq 0$ with equality if and only if $\rho = \sigma$ is still true in this setting, but it is no longer a direct application of Jensen's inequality. Indeed, not every classical entropy inequality is valid in the quantum setting.  For example, a marginal of a classical joint probability distribution never has an entropy 
exceeding that of the joint distribution itself. In the quantum setting, however, this is false.

Even when a classical entropy inequality has a valid quantum analog, its proof in the quantum setting may be much more difficult. Probably the first conjecture about a classical entropy inequality that did not obviously hold true in quantum mechanics was made by Lanford and Robinson, namely, Strong Subadditivity of the von Neumann entropy (SSA). This was proved in \cite{LR2} based on a convexity theorem for certain trace functionals in \cite{L}. It is a remarkable fact that while some of the inequalities in classical information theory carry through in quantum information theory and some do not,  it is often the more complicated ones that carry through, SSA  for instance, while some of the simpler ones fail. SSA is essentially equivalent to the joint convexity of the Umegaki relative entropy, and the quantum analog of  \eqref{clDPI}, the Data Processing Inequality. Wehrl's paper \cite{W} provides a  good review of the what was known on entropy and relative entropy inequalities, classical and quantum,  
through  1978, and it is an especially good source on differences between the classical and quantum settings. 

The {\em quantum Data Processing Inequality for a divergence $D$} says that for completely positive trace preserving maps
$\mathcal E$ and all density matrices $\rho$ and $\sigma$, one has
\begin{equation}\label{qDPI}
D(\mathcal E\rho ||\mathcal E \sigma) \leq D(\rho ||\sigma)\ .
\end{equation}
Completely positive trace preserving (CPTP) maps are also known as quantum operations, and are the general class of state transformations possible in an open quantum system \cite{Kr}.  Applying any such operation can only make the states harder to distinguish, and hence if $D$ is to have an operational meaning, it must satisfy the Data Processing Inequality.
% A good review of the situation until 1978 is in Wehrl's article \cite{W}.
%Our goal here is to bring together the historical background and the present status of the subject, especially in the context of a 2013 conjecture of Audenaert and Datta \cite{AD0}, that we discuss shortly

While the von Neumann entropy \eqref{vnent} is the natural analog of the classical entropy $S(\rho) = - \int \rho \log \rho$, the situation is more complicated when one considers relative entropy and other divergences, since these involve two density matrices $\rho$ and $\sigma$ that need not commute. Thus, there are many ways one might try to write down a quantum analog of $\rho \log (\rho/\sigma)$, for example. 

Relative entropy arises in the answers to all sorts of questions in classical probability theory, and it is not evident that it is Umegaki's quantum generalization that answers the corresponding question in the quantum setting. Among the many other expressions that reduce to it when $\rho$ and $\sigma$ commute are
\begin{equation}\label{put}
\tr[\rho \log (\sigma^{-1/2}\rho \sigma^{-1/2})] \qquad{\rm and}\qquad  
\tr[\rho \log (\rho^{1/2}\sigma^{-1} \rho^{1/2})]\ .
\end{equation}
The Data Processing Inequality may be used to winnow the large field of putative divergences leaving a reduced class for which one might hope to find operational meanings. For example, the trace function on the right in \eqref{put} does satisfy \eqref{qDPI}, as does the Umegaki relative entropy, but the trace function on the left does not \cite{CL3}. 
It turns out  that the Umegaki relative entropy is the ``right choice'' as far as the direct quantum analog of the classical decision problem that was discussed above  is concerned, as shown by Hiai and Petz \cite{HP91}:  Consider two density matrices $\rho$ and $\sigma$ on a finite dimensional Hilbert space $\cH$. Let $\mathcal{P}(\cH^{\otimes N})$ be the set of orthogonal projections onto subspaces of $\cH^{\otimes N}$.  Fix $\epsilon>0$, and define
\begin{equation}\label{qdec}
\beta_{\epsilon,N}(\rho,\sigma) = \inf \left\{ \log \tr[\sigma^{\otimes N}A_N] 
\ :\  A_N\in  \mathcal{P}(\cH^{\otimes N})\ , \
\tr[\rho^{\otimes N}A_N] > 1-\epsilon\ 
\right\} .
\end{equation}
The result of Hiai and Petz \cite{HP91} is that then \eqref{uplow} is valid also in this quantum setting where now $D(\rho||\sigma)$ denotes the Umegaki relative entropy. (They actually prove this in a somewhat more general setting.) Another of their results is that 
\begin{equation}\label{BS}
 \tr[\rho (\log \rho - \log \sigma)] \leq \tr[\rho \log (\rho^{1/2}\sigma^{-1} \rho^{1/2})]\ ,
 \end{equation}
 and the inequality is strict when $\rho$ and $\sigma$ do not commute. Hence while the quantity on the right in \eqref{BS} may seem to be a natural extension of the notion of relative entropy to the quantum setting, and while it does satisfy \eqref{qDPI}, it is not the one that is relevant to the decision problem that we have been discussing -- which is not to say that it is not the relevant quantity for some other problem.
 
Thus, when generalizing classical entropy inequalities to the quantum setting, there is the difficulty that non-commutativity prevents one from directly mimicking the classical proofs, but also the non-commutativity is the source of a multiplicity of distinct quantum analogs of classical entropies or divergences: Which analog is meaningful in which settings?

This situation naturally arises when one considers  quantum analogs of the R\'enyi relative entropies. The obvious generalization  of \eqref{aren}, namely,
\begin{equation}\label{arenQ}
D_\alpha(\rho ||\sigma) := \frac{1}{\alpha-1}\log \left(\tr[\rho^\alpha \sigma^{1-\alpha}]\right)\ ,
\end{equation}
turns out to have the same operational meaning that was given for \eqref{aren} by Ciszl\'ar; this was proved by Mosonyi and Hiai \cite{MosHiai}.
However,  \eqref{arenQ} is not the only quantum analog of \eqref{aren} to have an operational meaning. 

Another generalization of the relative R\'enyi entropy was introduced recently by 
M\"uller-Lennert, Dupuis, Szehr, Fehr, Tomamichel \cite{MDSFT} and Wilde, Winter, Yang \cite{WWY}.  They introduces the {\em sandwiched R\'enyi entropies} 
\begin{equation}\label{arenQS}
\widetilde{D}_\alpha(\rho ||\sigma) := \frac{1}{\alpha-1}\log \left(\tr[(\sigma^{(1-\alpha)/2\alpha} \rho \sigma^{(1-\alpha)/2\alpha})^\alpha]\right)\ ,
\end{equation}  
For certain values of the the parameter, an operational meaning has been given in \cite{MosOg}.

Audenaert and Datta realized that all of these different quantum R\'enyi entropies -- and more -- could be brought together  in a two parameter family, the $\alpha-z$ R\'enyi entropies defined by 
\begin{equation}\label{azrenQ}
D_{\alpha,z}(\rho ||\sigma) := \frac{1}{\alpha-1}\log \left(\tr[(\sigma^{(1-\alpha)/2z} \rho^{\alpha/z} \sigma^{(1-\alpha)/2z})^z]\right)\ .
\end{equation}
Evidently, $D_\alpha(\rho ||\sigma) = D_{\alpha,1}(\rho ||\sigma)$ and 
$\widetilde{D}_\alpha(\rho ||\sigma)  = D_{\alpha,\alpha}(\rho ||\sigma)$.  The $\alpha$-$z$ R\'enyi relative entropies have appeared earlier in a paper by Jaksic, Ogata, Pautrat and Pillet \cite{JOPP}.

Audenaert and Datta raised the question: 
\begin{align*}
&\text{For which values of $\alpha$ and $z$ does $D_{\alpha,z}$ satisfy the}\\ & \text{quantum Data Processing Inequality \eqref{qDPI}?}
\end{align*}
  
%Alternatively: for which $\alpha$ and $z$ is
%$(\rho,\sigma) \mapsto D_{\alpha,z}(\rho||\sigma)$ jointly convex?
By an argument of Lindblad and Uhlmann, based on Stinespring's theorem, this question is equivalent to whether the 
trace functionals appearing in the $D_{\alpha,z}$ have certain convexity or concavity properties. This will be explained below in detail, but in a nutshell, for our finite dimensional setting, the Stinespring Representation Theorem \cite{St} allows one to express CPTP operations in terms of isometric injections, unitary conjugations and partial traces, in such a way that once one knows the monotonicity \eqref{qDPI}  under partial traces -- one example of a CPTP map -- one knows it in general. 
Since partial traces can be written as averages over unitary conjugations \cite{U73}, it suffices to prove $D(\mathcal{E} \rho||\mathcal{E} \sigma) \leq D(\rho||\sigma)$ when $\mathcal{E}$ takes the form
$$\mathcal{E}(\rho) = \frac{1}{M}\sum_{j=1}^M U_j \rho U_j^*\ $$
where $U_1,\dots,U_M$ are unitary. Now consider, for example, the trace functional in \eqref{arenQ}, $\tr[\rho^\alpha \sigma^{1-\alpha}]$. Evidently, for each $j$,
$$\tr[((U_j\rho U_j^*)^\alpha (U_j\sigma U_j^*)^{1-\alpha}]  = \tr[\rho^\alpha \sigma^{1-\alpha}]\ ,$$
and therefore if $(\rho,\sigma) \mapsto \tr[\rho^\alpha \sigma^{1-\alpha}]$ is jointly concave, the standard quantum R\'enyi relative entropy \eqref{arenQ}  with $\alpha \in (0,1)$ satisfies 
$D(\mathcal{E} \rho||\mathcal{E} \sigma) \leq D(\rho||\sigma)$. 

However, the joint concavity of  $(\rho,\sigma) \mapsto \tr[\rho^\alpha \sigma^{1-\alpha}]$ is a special case of the Lieb Concavity Theorem \cite{L}, which says that for $0 \leq p,q$ with $p+q \leq 1$, $(A,B) \mapsto \tr[A^pB^q]$ is jointly concave. One may also consider  R\'enyi relative entropies for $\alpha > 1$. In this case, the factor of $\frac{1}{1-\alpha}$ in \eqref{arenQ} is negative, and so the DPI will follow if $(\rho,\sigma) \mapsto \tr[\rho^\alpha \sigma^{1-\alpha}]$ is {\em convex}, and for  $\alpha\in (1,2)$
this is true by  the Ando Convexity Theorem \cite{A}, which says that  $(A,B) \mapsto \tr[A^pB^q]$ is jointly convex for $-1 \leq q \leq 0$ and $0 \leq q+p \leq 1$. (Since we shall be concerned with the relation between convexity and concavity inequalities in what follows, 
it is worth remarking, that the proof of Ando's Convexity Theorem relies  on  the concavity result of \cite{L}.)  Since the Umegaki relative entropy is the $\alpha\to 1$ limit of the R\'enyi relative entropy, these convexity results also yield the Data Processing Inequality for it as well, as first proved by Lindblad \cite{L74,L75}. 

To prove the data processing inequality for the $\alpha-z$ R\'enyi entropies one requires concavity and convexity of more complicated trace functionals. ( In fact, as explained below, concavity/convexity is necessary as well as sufficient for the data processing inequality.)
The convexity theorem that implies the data processing inequality for certain of the sandwiched R\'enyi entropies was proved in \cite{FL} (see also \cite{Be}).
In their paper \cite{AD} Audenaert and Datta deduced the data processing inequality in certain cases 
 from concavity theorems of Hiai and conjectured a precise parameter regime for the validity of the data processing inequality.  
 
Let  $\mathcal M_N$ denote  the set of complex $N\times N$ matrices and  $\mathcal P_N$ denote the subset of positive definite matrices.
 Their conjecture refers to the trace function at the heart of the definition of the $\alpha-z$ R\'enyi entropies, and may be stated as follows:

\begin{conjecture}\label{conj}
If $1\leq p\leq 2$ and $-1\leq q<0$, then for any $K\in\mathcal M_N$
$$
\mathcal P_N\times\mathcal P_N \ni (A,B) \mapsto \tr \left( B^{q/2} K^* A^p K B^{q/2} \right)^{1/(p+q)}
$$
is jointly convex.
\end{conjecture}

Currently, this conjecture is known for $q=p-1$ \cite{A}, for $p=1$ \cite{FL} and for $p=2$ \cite{CFL}. It is open in the remaining cases.

Conjecture \ref{conj} is stated in a different, but equivalent form in \cite{AD}, namely, as convexity of $\mathcal P_N \ni A \mapsto \tr \left( A^{q/2} K^* A^p K A^{q/2} \right)^{1/(p+q)}$. The equivalence of this version with the one stated above will be explained at the end of this introduction.

The concavity of $(A,B) \mapsto \tr[A^pB^q]$ with $0 \leq p,q$ and $p+q \leq 1$ lies at the basis of all of the examples  of quantum data processing inequalities that we have discussed. A number of different proofs of this fundamental result were given by other authors, and the two that are most relevant for our purposes appeared within a few years of \cite{L}. One of these was due to Araki \cite{A75,A76}. His proof introduced a method making use of the {\em relative modular operator} and other tools coming from modern theory of von Neumann algebras, and has the merit of providing generalizations of the inequality to this setting. The paper \cite{A75} is especially clear on the matricial case and is accessible without prior knowledge of the the modern theory of von Neumann algebras. An even more pedestrian account of this approach, which is now well known in the quantum information theory community,  can be found in \cite{NP}.   

Another methodology was introduced by Epstein \cite{Ep}, who employs the theory of Herglotz functions, also known as Pick or Nevanlinna functions.  These are functions
$\varphi$ that are analytic in the open upper half plane, $\C_+$, and that have a positive imaginary part.  Such functions have a canonical integral representation, as recalled below, and from this a number of their properties follow. These functions have played a central role in the theory of operator inequalities since the 1934 theorem of L\"owner \cite{Low} stating that a function $f:\R \to \R$ is such that for all self adjoint matrices $A,B$ of the same size, $f(A) \geq f(B)$ whenever $A\geq B$ (in the usual ordering of self-adjoint matrices) if and only if $f$ has an analytic continuation to a Herglotz function.   The deep part of L\"owner's Theorem is the ``only if'' part:  It is an immediate consequence of the integral representation for Herglotz functions that if $f$ is Herglotz, not only is $A\mapsto f(A)$ operator monotone, it is operator concave. A full account of this theory, with three distinct proofs of L\"owner's Theorem, can be found in the book of Donoghue \cite{Don}. A forthcoming book of Simon \cite{S} will present 11 distinct proofs, three of them new. 

One of the merits of Epstein's method, explained in Section 3,  is that he was able to use it to prove a conjecture that had been made in \cite{L}, namely that for fixed self adjoint $B$, the map
$A \mapsto \tr[(BA^pB)^{1/p}]$ is concave for $p\in (0,1)$. At the end of the first paragraph of his paper, Epstein wrote: ``The applicability
of the method obviously extends beyond the examples treated here.''  Hiai \cite{Hi0,Hi1} has carried out a thorough development of the method, bringing in a number of significant new ideas. However, the method still has limitations.   It was conjectured in \cite{CL99} and proved in \cite{CL2} that  $A \mapsto \tr[(BA^pB)^{1/p}]$ is convex for 
$1 < p < 2$. Since $z\mapsto z^p$ is  a Herglotz function only  for $0 \leq p \leq 1$, Epstein's argument cannot be adapted to this case.  The method of proof, which turns out to be useful also for some cases of the Audenaert-Datta Conjecture, relied on variational arguments, specifically Legendre transforms, and, most crucially, the fact that if $f(x,y)$ is jointly convex on $\R^m\times \R^m$, then $g(y) := \inf_{x\in \R^m}\{f(x,y)\}$ is convex in $y$.  This variational method is explained in Section 4.

One of the intriguing aspects of the story of progress on the Audenaert-Datta Conjectures  is the interplay between the ``analyticity method" and the ``variational method''.
It appears that the analyticity method is especially useful for proving concavity and the variational method  is more useful for proving convexity, but this is not meant to be an absolute distinction. However, it is our belief that understanding these two methodologies is a worthwhile endeavor also for work outside the Audenaert--Datta program. 
There  are quite a number of papers \cite{Hi0,CL2,Hi1,FL,AD,CFL,Hi2} over which trace inequalities related to the Audenaert and Datta  conjecture are spread out. 
These papers, both before and after the Audenaert--Datta paper \cite{AD}, are often not self-contained, and one of our goals is to try to tell the story seamlessly and thereby make the results from the  literature more easily available also to newcomers to the field. There are no new results in this paper, but we hope to present a newly coherent account of  recent advances.

%\texttt{Wehrl's paper}
%
%
%
%
%\texttt{We need to add here a lot of words. The message should be that thereWe tried to collect them together in one place. There are no new results here.}
%
%
%
%Motivated by a problem about entropic quantities in quantum information theory, which we recall in the next section, Audenaert and Datta made the following conjecture \cite{AD}.

In the following we will study a more general problem than the one occurring in Conjecture 1, with three parameters instead of two. For $A,B\in\mathcal P_N$, $K\in\mathcal M_N$ and parameters $p,q,s\in\R$ we define
$$
\Psi_{p,q,s}(A,B) := \tr \left( B^{q/2} K^* A^p K B^{q/2} \right)^s \,,
$$
and the problem is to determine the values of $p,q,s$ such that $\Psi_{p,q,s}$ is jointly convex or concave. 
Because of symmetries, it suffices to consider $p,q,s$ such that 
\begin{equation}\label{pqs}
p\geq q
\qquad\text{and}\qquad
s>0 \,.
\end{equation}
To see this, note that by an approximation argument, we can assume that $K$ is invertible, and then
$\left( B^{q/2} K^* A^p K B^{q/2} \right)^s = \left( B^{-q/2} K^{-1} A^{-p} K^{-*} B^{-q/2} \right)^{-s}$, so that with $K^{-*}$ replacing $K$ on the right,
$$
\Psi_{p,q,s}(A,B) =  \Psi_{-p,-q,-s}(A,B)\ .$$
Next since 
 $B^{q/2} K^* A^p K B^{q/2}$ and $A^{p/2} K B^q K^* A^{p/2}$ have the same non-zero eigenvalues with the same mutiplicities, with $K^*$ replacing $K$ on the right,
\begin{equation}
\label{eq:commute}
\Psi_{p,q,s}(A,B) = \Psi_{q,p,s}(B,A).
\end{equation}

\vfill\break

The following theorem, which is the main subject of this paper, summarizes our current knowledge about concavity and convexity properties of $\Psi_{p,q,s}$.

\begin{theorem}\label{main}
Let $K\in\mathcal M_N$ be arbitrary.
\smallskip
\begin{enumerate}
\item If $0\leq q\leq p\leq 1$ and $0<s\leq 1/(p+q)$, then $\Psi_{p,q,s}$ is jointly concave.
\smallskip
\item If $-1\leq q\leq p\leq 0$ and $s>0$, then $\Psi_{p,q,s}$ is jointly convex.
\smallskip
\item If $1\leq p< 2$, $-1\leq q\leq 0$ and $s\geq\min\{1/(p-1),1/(q+1)\}$, then $\Psi_{p,q,s}$ is jointly convex. If $p=2$, $-1\leq q\leq 0$ and $s\geq 1/(2+q)$, then $\Psi_{p,q,s}$ \smallskip
is jointly convex.
\end{enumerate}
\end{theorem}

In part (3), for $p=1$ the condition $s\geq\min\{1/(p-1),1/(q+1)\}$ is to be understood as $s\geq 1/(q+1)$.  The information contained in the thoerem, and extended by symmetry in $p$ and $q$,  is summarized in the following figure:

\medskip

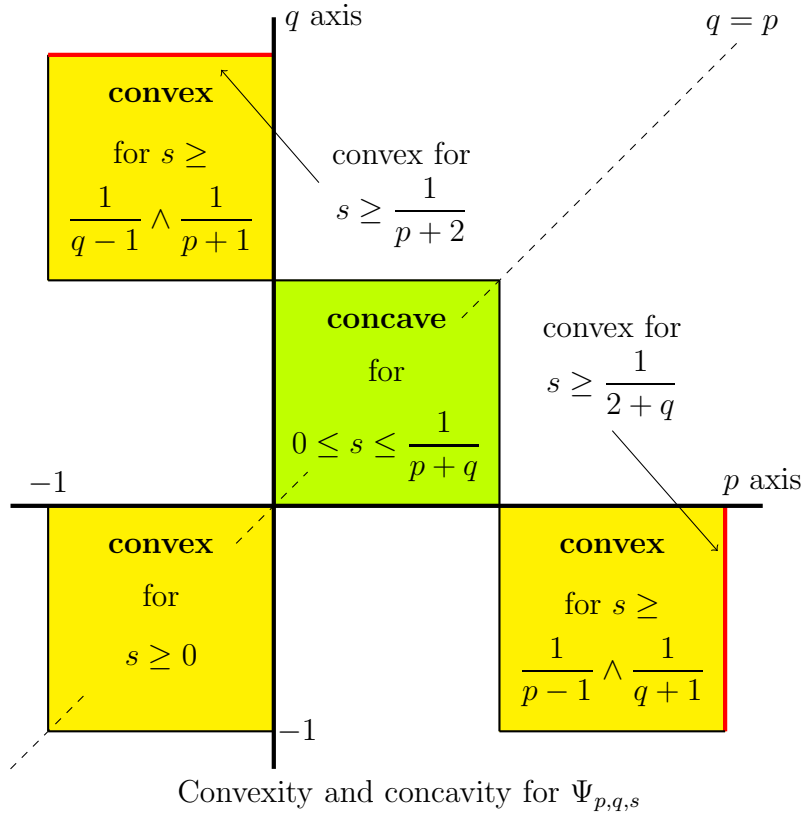
\begin{center}
\begin{tikzpicture}

\draw [fill=yellow,yellow] (0.5,0) rectangle (3.5,3);
\draw [thick] (0.5,0) -- (0.5,3);
\draw [thick] (3.5,0) -- (3.5,3);
\draw [thick] (0.5,0) -- (3.5,0);
\draw [thick] (0.5,3) -- (3.5,3);
\node [right] at (3.4,0) {{$-1$}};
\node [above] at (0.5,3) {{$-1$}};

\draw [fill=lime,lime] (3.5,3) rectangle (6.5,6);
\draw [thick] (3.5,3) -- (3.5,6);
\draw [thick] (6.5,3) -- (6.5,6);
\draw [thick] (3.5,3) -- (6.5,3);
\draw [thick] (3.5,6) -- (6.5,6);

\draw [fill=yellow,yellow] (6.5,0) rectangle (9.5,3);
\draw [thick] (6.5,0) -- (6.5,3);
\draw [ultra thick, red] (9.5,0) -- (9.5,3);
\draw [thick] (6.5,0) -- (9.5,0);
\draw [thick] (6.5,3) -- (9.5,3);
\draw [->] (8,4) -- (9.4,2.4);
\node [above] at (8,5.1) {{convex for}};
\node [above] at (8,4) {{${\displaystyle s \geq \frac{1}{2+q}}$}};

\draw [fill=yellow,yellow] (0.5,6) rectangle (3.5,9);
\draw [thick] (0.5,6) -- (0.5,9);
\draw [thick] (3.5,6) -- (3.5,9);
\draw [thick] (0.5,6) -- (3.5,6);
\draw [ultra thick, red] (0.5,9) -- (3.5,9);
\draw [->] (4.1,7.3) -- (2.8,8.8);
\node [above] at (5.2,7.4) {{convex for}};
\node [above] at (5.2,6.3) {{${\displaystyle s \geq \frac{1}{p+2}}$}};

\draw [ultra thick] (3.5,-0.5) -- (3.5,9.5);
\draw [ultra thick] (0,3) -- (10,3);

\node [right] at (3.5,9.5) {{$q$ axis}};
\node [above] at (10,3) {{$p$ axis}};
\node [right] at (9.1,9.4) {{$q=p$ }};

\node [below] at (5,5.7) {{\bf concave}};
\node [below] at (5,5.1) {{for }};
\node [below] at (5,4.4) {{${\displaystyle 0 \leq s  \leq \frac{1}{p+q}}$ }};

\node [below] at (2,2.7) {{\bf convex}};
\node [below] at (2,2.1) {{for }};
\node [below] at (2,1.3) {{$s \geq 0$ }};

\node [below] at (8,2.7) {{\bf convex}};
\node [below] at (8,2.0) {{for $s\geq$}};
\node [below] at (8,1.4) {{${\displaystyle  \frac{1}{p-1}\wedge \frac{1}{q+1}}$ }};

\node [below] at (2,8.7) {{\bf convex}};
\node [below] at (2,8) {{for $s\geq$}};
\node [below] at (2.0,7.4) {{${\displaystyle  \frac{1}{q-1}\wedge \frac{1}{p+1}}$ }};

\draw [dashed] (0,-0.5) -- (1,0.5);
\draw [dashed] (3,2.5) -- (3.95,3.45);
\draw [dashed] (6,5.5) -- (9.7,9.2);

\node [below] at (5.3,-.5) {Convexity and concavity for $\Psi_{p,q,s}$};

\end{tikzpicture}

\end{center}

\medskip

\emph{History.} We discuss the three cases in the theorem separately, plus the case $s=1$ which historically came first and played an important role in the development of the field of matrix analysis.
\begin{enumerate}
\item[(0)] Case $s=1$: This is due to Lieb \cite{L} for $0\leq q\leq p\leq 1$ with $p+q\leq 1$, as well as for $-1\leq q\leq p\leq 0$, and due to Ando \cite{A} for $1\leq p\leq 2$, $-1\leq q\leq 0$ with $p+q\geq 1$.
\smallskip
\item Case $0\leq q\leq p\leq 1$: Partial results by Hiai for $1\leq s\leq 1/(p+q)$ \cite[Theorem 2.3]{Hi0} and for $1/2\leq s\leq 1/(p+q)$ \cite[Theorem 2.1 (1)]{Hi1} and by the authors for $0<s\leq 1/(1+q)$ \cite[Theorem 4.4]{CFL} (see also \cite[Proposition 3]{FL} for $p=1$, $s=1/(1+q)$). Later Hiai found a proof \cite[Theorem 2.1]{Hi2} that covers the full range. We reproduce this proof in Section \ref{sec:epstein}.
\smallskip
\item Case $-1\leq q\leq p\leq 0$: This is due to Hiai. After a partial result for $1/2\leq s\leq -1/(p+q)$ \cite[Theorem 2.1 (2)]{Hi1}, the full result appears in \cite[Theorem 2.1]{Hi2}. We reproduce this proof in Section \ref{sec:var}.  The key to Hiai's proof is to consider the equivalent result for {\em negative} $s$: Of the two equivalent results that are related by changing the signs of $p$, $q$ and $s$, {\em only one is amenable to treatment by the variational method}. 
\smallskip
\item Case $1\leq p\leq 2$, $-1\leq q\leq 0$: Result for $p=1$, $s=1/(1+q)$ due to the last two authors \cite[Proposition 3]{FL}. Remaining results due to the authors \cite{CFL}. We reproduce the proofs in Section \ref{sec:var}.
\end{enumerate}

\medskip

We now complement Theorem \ref{main} with necessary conditions.

\begin{proposition}\label{mainnec}
Let $s>0$ and $p\geq q$ with $(p,q)\neq(0,0)$.
\begin{enumerate}
\item If $\mathcal P_2\times\mathcal P_2 \ni (A,B)\mapsto \Psi_{p,q,s}(A,B)$ is jointly concave for $K=1$, then $0\leq q\leq p\leq 1$ and $0<s\leq 1/(p+q)$.
\smallskip
\item If $\mathcal P_4\times\mathcal P_4 \ni (A,B)\mapsto \Psi_{p,q,s}(A,B)$ is jointly convex for $K=1$, then either $-1\leq q\leq p\leq 0$ and $s>0$ or $1\leq p\leq 2$, $-1\leq q\leq 0$, $(p,q)\neq (1,-1)$ and $s\geq 1/(p+q)$.
\end{enumerate}
\end{proposition}

This proposition is due to Hiai \cite[Propositions 5.1(2) and 5.4(2)]{Hi1}. It is natural to conjecture that these necessary conditions are also sufficient.

\begin{conjecture}
If $1\leq p\leq 2$, $-1\leq q<0$ and $s\geq 1/(p+q)$, then for any $K\in\mathcal M_N$
$$
\mathcal P_N\times\mathcal P_N \ni (A,B) \mapsto \tr \left( B^{q/2} K^* A^p K B^{q/2} \right)^s
$$
is jointly convex.
\end{conjecture}

Note that for $s=1/(p+q)$ this is Conjecture \ref{conj} of Audenaert and Datta. The  remaining case of the conjecture is
$$
1<p<2 \,,\ -1\leq q<0 \,, \ 1/(p+q)\leq s < \min\{1/(p-1),1/(q+1)\}\,.
$$
(In fact, the case $s=1$ can be excluded as a conjecture, due to a theorem of Ando.)

%%%%%%%%%%%%%%%%%%%%%%%%%%%%%%%%

\bigskip

We now turn to a different, but related problem. For $A\in\mathcal P_N$, $K\in\mathcal M_N$ and parameters $p,s\in\R$ we define
$$
\Upsilon_{p,s}(A) := \tr (K^* A^p K)^s \,,
$$
that is, $\Upsilon_{p,s}(A)=\Psi_{p,q,s}(A,1)$ for any $q$. As before, we can and will restrict ourselves to the case $s>0$.

We will be interested in convexity and concavity properties of $\Upsilon_{p,s}$. While those are a consequence of similar properties of $\Psi_{p,q,s}$, we will conversely prove them first and use them in our discussion of $\Psi_{p,q,s}$.

\begin{proposition}\label{ups}
Let $K$ be arbitrary.
\smallskip
\begin{enumerate}
\item If $0\leq p\leq 1$ and $0<s\leq 1/p$, then $\Upsilon_{p,s}$ is concave.
\smallskip
\item If $-1\leq p\leq 0$ and $s>0$, then $\Upsilon_{p,s}$ is convex.
\smallskip
\item If $1\leq p\leq 2$ and $s\geq 1/p$, then $\Upsilon_{p,s}$ is convex.
\smallskip
\end{enumerate}
\end{proposition}

\emph{History.}
\begin{enumerate}
\item Case $0\leq p\leq 1$: Initial result for $s=1/p$ due to Epstein \cite{Ep}.  The first and third authors proved the result under the assumption $s\geq 1$ \cite[Theorem 1.1]{CL2}. The full result is due to Hiai \cite[Theorem 4.1 (1)]{Hi1}.
\smallskip
\item Case $-1\leq p<0$: Result due to Hiai \cite[Theorem 4.1 (2)]{Hi1}.
\smallskip
\item Case $1\leq p\leq 2$: Result due to the first and third author \cite[Theorem 1.1]{CL2}.
\end{enumerate}

We next show that the conditions in Proposition \ref{ups} are also necessary.

\begin{proposition}\label{upsnec}
Let $s>0$ and $p\neq 0$.
\smallskip
\begin{enumerate}
\item If $\mathcal P_2\ni A\mapsto \Upsilon_{p,s}(A)$ is concave for any invertible $K$, then $0<p\leq 1$.
\smallskip
\item If  $\mathcal P_4\ni A\mapsto \Upsilon_{p,s}(A)$ is convex for any invertible $K$, then either $-1\leq p <0$ and $s>0$ or $1\leq p\leq 2$ and $s\geq 1/p$.
\end{enumerate}
\end{proposition}

This proposition is due to Hiai \cite[Propositions 5.1(1) and 5.4(1)]{Hi1}. Earlier, the first and third authors \cite{CL2} had used a similar argument to show that $\Upsilon_{p,s}$ is not convex or concave for $p>2$. Related arguments also appear in \cite{B}.

\medskip

We conclude this introduction by  explaining why Conjecture \ref{conj} is equivalent to the form in \cite{AD}, which corresponds to taking $A=B$ in our form. Given $A,B\in\mathcal P_N$ and $K\in\mathcal M_N$ let
$$
C_{A,B} = \begin{pmatrix}
A & 0 \\ 0 & B
\end{pmatrix} \in\mathcal P_{2N}
\qquad\text{and}\qquad
L_K = \begin{pmatrix}
0 & K \\ 0 & 0
\end{pmatrix} \in\mathcal M_{2N} \,.
$$
Then
$$
\tr \left( C_{A,B}^{q/2} L_K^* C_{A,B}^p L_K C_{A,B}^{q/2} \right)^s
= \tr \left( B^{q/2} K^* A^p K B^{q/2} \right)^s \,.
$$
Thus, since $(A,B)\mapsto C_{A,B}$ is linear, convexity of $C\mapsto \tr (C^{q/2}L^* C L^p C^{q/2})^s$ on $\mathcal P_{2N}$, implies joint convexity of $(A,B)\mapsto \tr \left( B^{q/2} K^* A^p K B^{q/2} \right)^s$ on $\mathcal P_N\times\mathcal P_N$.

%%%%%%%%%%%%%%%%%%%%%%%%%

\section{Application: Monotonicity of the $\alpha-z$ relative entropies}

In this section we present an application of Theorem \ref{main} to a problem motivated by quantum information theory. For $\rho,\sigma\in\mathcal P_N$ and $\alpha,z>0$ with $\alpha\neq 1$, we consider the so-called $\alpha-z$-relative R\'enyi entropies
$$
D_{\alpha,z}(\rho\|\sigma) = \frac1{\alpha-1} \ln \frac{\tr\left( \sigma^{(1-\alpha)/(2z)} \rho^{\alpha/z} \sigma^{(1-\alpha)/(2z)}\right)^z}{\tr\rho} \,.
$$
(One also obtains interesting quantities by taking the limit $\alpha\to 1$, possibly simultaneously with $z\to 0$ \cite{AD}, but for the sake of brevity we exclude this case.) These functionals appeared in
\cite[Sec. 3.3]{JOPP} and were further studied  in
\cite{AD}, where the question was raised
whether the $\alpha-z$-relative R\'enyi entropies are monotone under
completely positive, trace preserving maps (CPTP), that is, whether for any such map $\mathcal E$ and any $\rho,\sigma\in\mathcal P_N$ one has
$$
D_{\alpha,z}(\rho\|\sigma) \geq D_{\alpha,z}(\mathcal E(\rho)\|\mathcal E(\sigma)) \,.
$$

\begin{proposition}\label{monoequi}
Let $\alpha,z>0$ with $\alpha\neq 1$ and define
$$
p=\alpha/z
\qquad\text{and}\qquad
q=(1-\alpha)/z \,.
$$
Then $D_{\alpha,z}$ is monotone under CPTP maps on $\mathcal P_N$ for all $N$ if and only if $\Psi_{p,q,1/(p+q)}$ with $K=1$ is jointly convex (if $\alpha>1$) or jointly concave (if $\alpha<1$) on $\mathcal P_N\times\mathcal P_N$ for all $N$.
\end{proposition}

\begin{proof}
For simplicity of notation, we write $\Psi:=\Psi_{p,q,1/(p+q)}$ with $K=1$ in the following. Clearly, $D_{\alpha,z}$ is monotone under CPTP maps if and only if $\Psi$ is monotone decreasing (if $\alpha>1$) or increasing (if $\alpha<1$) under CPTP maps. Therefore the proposition follows from what we prove in the following two steps.

\emph{Step 1.} We show that for $\alpha>1$, if $\Psi$ is monotone decreasing under CPTP maps on $\mathcal P_{2N}$, then $\Psi$ is jointly convex on $\mathcal P_N\times\mathcal P_N$. A similar assertion holds for $\alpha<1$ if decreasing and convex are replaced by increasing and concave.

We give the proof for $\alpha>1$. For $\rho_0,\rho_1,\sigma_0,\sigma_1\in\mathcal P_N$ and $\theta\in[0,1]$ we consider the operators
$$
\rho = (1-\theta)\rho_0\otimes |\uparrow\rangle\langle\uparrow| + \theta \rho_1\otimes |\downarrow\rangle\langle\downarrow|
\qquad\text{and}\qquad
\sigma = (1-\theta)\sigma_0\otimes |\uparrow\rangle\langle\uparrow| + \theta \sigma_1\otimes |\downarrow\rangle\langle\downarrow|
$$
on $\C^N\otimes\C^2$, where $|\uparrow\rangle,|\downarrow\rangle$ denote a basis of $\C^2$. Applying the assumed monotonicity with the channel that takes the partial trace over $\C^2$, we conclude that
$$
\Psi((1-\theta)\rho_0+\theta\rho_1,(1-\theta)\sigma_0+\theta\sigma_1) \leq (1-\theta) \Psi(\rho_0,\sigma_0) + \theta \Psi(\rho_1,\sigma_1) \,,
$$
which means that $\Psi$ is jointly convex.

\emph{Step 2.} We show that for $\alpha>1$, if $\Psi$ is jointly convex on $\mathcal P_{N^2}\times\mathcal P_{N^2}$, then $\Psi$ is monotone decreasing under CPTP maps on $\mathcal P_N$. A similar assertion holds for $\alpha<1$ if convex and decreasing are replaced by concave and increasing.

Again, we assume $\alpha>1$. Following a method of Lindblad and Uhlmann, we use Stinespring's theorem \cite{St} to obtain an integer $N'\leq N^2$, a non-negative matrix $\tau\in\mathcal M_{N'}$ with $\tr\tau=1$ (which can be chosen to be rank one) and a unitary $U$ on $\C^N\otimes\C^{N'}$ such that
$$
\mathcal E(\gamma) = \tr_2 U \left(\gamma\otimes \tau \right) U^* \,.
$$
Thus, if $du$ denotes normalized Haar measure on all unitaries on $\C^{N'}$, then
\begin{equation}
\label{eq:uhlmann}
\mathcal E(\gamma) \otimes (N')^{-1} 1_{\C^{N'}} = \int (1\otimes u) U \left(\gamma\otimes \tau \right) U^* (1\otimes u^*) \,du \,.
\end{equation}
By the tensor property of $\Psi$,
\begin{equation} \label{two}
\Psi(\mathcal E(\rho),\mathcal E(\sigma))
= \Psi(\mathcal E(\rho)\otimes (N')^{-1} 1_{\C^{N'}}\,, \mathcal E(\sigma)\otimes (N')^{-1} 1_{\C^{N'}}) \,.
\end{equation}
By \eqref{eq:uhlmann} and the assumed convexity the expression in \eqref{two} is bounded from above by
\begin{align*}
& \int \Psi( (1\otimes u) U \left(\rho\otimes \tau \right) U^* (1\otimes u^*) , (1\otimes u) U \left(\sigma\otimes \tau \right) U^* (1\otimes u^*) ) \,du \,.
\end{align*}
By unitary invariance and the tensor property of $\Psi$, the integrand here equals
$$
\Psi( (1\otimes u) U \left(\rho\otimes \tau \right) U^* (1\otimes u^*), (1\otimes u) U \left(\sigma\otimes \tau \right) U^* (1\otimes u^*) )
= \Psi( \rho\otimes \tau, \sigma\otimes \tau )
= \Psi( \rho, \sigma ) \,,
$$
and therefore, recalling that Haar measure is normalized, we conclude that
$$
\Psi(\mathcal E(\rho),\mathcal E(\sigma)) \leq \Psi( \rho, \sigma ) \,,
$$
which means that $\Psi$ is monotone decreasing under CPTP maps.
\end{proof}

Combining this proposition with Theorem \ref{main} and Proposition \ref{mainnec} we obtain the following monotonicity result.

\begin{corollary}
$D_{\alpha,z}$ is monotone under completely positive, trace preserving maps if one of the following holds
\begin{align*}
& 0<\alpha< 1 \ \text{and}\ z\geq\max\{\alpha,1-\alpha\} \,,\\
& 1<\alpha\leq 2 \ \text{and}\ z\in\{\alpha/2,1,\alpha\} \,,\\
& 2\leq\alpha\leq\infty \ \text{and}\ z=\alpha \,.
\end{align*}
Conversely, if $D_{\alpha,z}$ is monotone under completely positive, trace preserving maps, then one of the following holds
\begin{align*}
& 0<\alpha< 1 \ \text{and}\ z\geq\max\{\alpha,1-\alpha\} \,,\\
& 1<\alpha\leq 2 \ \text{and}\ \alpha/2\leq z\leq\alpha \,,\\
& 2\leq\alpha<\infty \ \text{and}\ \alpha-1\leq z\leq\alpha \,.
\end{align*}
\end{corollary}

If Conjecture \ref{conj} is true, then monotonicity holds also in the cases $1<\alpha\leq 2$, $\alpha/2<z<\alpha$, as well as $2\leq\alpha<\infty$, $\alpha-1\leq z<\alpha$, that is, in the full range allowed by the second part of the corollary.

\begin{remark} Hiai and Mosonyi have made further progress \cite{HM} on the conjecture under the additional assumption that the CPTP map in question is unital and preserves either $\rho$ or $\sigma$. 

\end{remark}

%%%%%%%%%%%%%%%%%%%%%%%%%

\section{The complex analysis method}\label{sec:epstein}

In this section we show part (1) of Theorem \ref{main}, namely that 
$$
(A,B) \mapsto \tr\left( B^{q/2} K^* A^p K B^{q/2} \right)^s
$$
is jointly concave for $0\leq q\leq p\leq 1$ and $0<s\leq 1/(p+q)$. Note that this also proves part (1) of Theorem \ref{ups}.

This was shown under the extra assumption $s\geq 1/2$ in \cite{Hi1} using the complex analysis method and under the extra assumption $s\leq 1/(1+q)$ in \cite[Theorem 4.4]{CFL} using the variational method. Here we follow Hiai's proof \cite[Theorem 2.1]{Hi2}, using the complex analysis method, which works for the full range of exponents.

\begin{proposition}\label{hiaiconc}
Let $0\leq p,q\leq 1$. Then
$$
(A,B) \mapsto \tr \left(1+ \left(B^{q/2} K^* A^p K B^{q/2}\right)^{-1/(p+q)}\right)^{-1} 
$$
is jointly concave.
\end{proposition}

Before proving this proposition we use it to deduce the concavity part of Theorem~\ref{main}.

\begin{proof}[Proof of Theorem \ref{main}. (1)]
Multiplying $A$ by a power of $t$ we deduce from Proposition \ref{hiaiconc} that
$$
(A,B) \mapsto \tr \left(1+ t \left(B^{q/2} K^* A^p K B^{q/2}\right)^{-1/(p+q)}\right)^{-1} 
$$
is jointly concave for any $t>0$. Multiplying by $t$ and letting $t\to\infty$ we deduce that
$$
(A,B)\mapsto \tr \left(B^{q/2} K^* A^p K B^{q/2}\right)^{1/(p+q)} 
$$
is jointly concave. This is the assertion for $s=1/(p+q)$.

Now let $\sigma:=s(p+q)<1$. Then
$$
x^\sigma = \frac{\sin(\pi\sigma)}{\pi} \int_0^\infty \frac{1}{1+t/x}\ t^{-1+\sigma}\,dt
\qquad\text{for all}\ x\geq 0
$$
and therefore
$$
\tr \left( B^{q/2} K^*A^p K B^{q/2} \right)^s = \frac{\sin(\pi\sigma)}{\pi} \int_0^\infty \tr\left( 1+ t \left( B^{q/2} K^* A^p K B^{q/2} \right)^{-1/(p+q)} \right)^{-1} t^{-1+\sigma}\,dt \,.
$$
As observed at the beginning of the proof, the integrand is jointly concave and therefore so is $\tr \left( B^{q/2} K^*A^p K B^{q/2} \right)^s$. This concludes the proof of part (1) of Proposition \ref{main}.
\end{proof}

We now turn to Hiai's proof of Proposition \ref{hiaiconc}. It is based on Epstein's method, which relies on the following lemma from complex analysis.

\begin{lemma}\label{epstein}
Let $\phi$ be analytic in $\C_+$ with $\im\phi\geq 0$ and assume that it extends continuously to a real function on $(R,\infty)$ for some $R>0$. Then
$$
\frac{d^2}{d\xi^2} \left( \xi \phi(1/\xi) \right) \leq 0
\qquad\text{for all}\ \xi\in (0,1/R) \,.
$$ 
\end{lemma}

\begin{proof}[Proof of Lemma \ref{epstein}]
By the Nevanlinna theorem \cite{Don} there are numbers $a\geq 0$ and $b\in\R$ and a non-negative measure $\mu$ with $\int_\R (t^2+1)^{-1}\,d\mu(t)<\infty$ such that
$$
\phi(z) = a z + b + \int_\R \left( \frac{1}{t-z} - \frac{t}{t^2+1} \right)\,d\mu(t)
\qquad\text{for all}\ z\in\C_+ \,.
$$
Since
$$
d\mu(t) = \pi^{-1} \text{w}-\lim_{\epsilon\to 0+} \im\phi(t+i\epsilon) dt
$$
and since $\phi$ is assumed real on $(R,\infty)$, we infer that $\supp\mu\subset(-\infty,R]$. Using this fact one can show that the above representation formula for $\phi(z)$ is also valid for $z\in (R,\infty)$ and, therefore,
$$
\xi\phi(1/\xi) = a + b\xi + \int_\R \left( \frac{\xi^2}{\xi t-1} - \frac{\xi t}{t^2+1} \right)\,d\mu(t)
\qquad\text{for all}\ \xi\in (0,1/R) \,.
$$
By dominated convergence it follows that $\xi\mapsto\xi\phi(1/\xi)$ is smooth on $(0,R)$ and that
$$
\frac{d^2}{d\xi^2} \left( \xi \phi(1/\xi) \right) = 2 \int_\R \frac{d\mu(t)}{(\xi t-1)^3}
\qquad\text{for all}\ \xi \in (0,1/R) \,.
$$
Since $(\xi t-1)^{-3}\leq 0$ for $\xi<1/R$ and $t\in\supp\mu\subset(-\infty,R]$, we obtain the claimed concavity.
\end{proof}

\begin{proof}[Proof of Proposition \ref{hiaiconc}]
By an approximation argument we may assume that $K$ is invertible. Let $C$ and $D$ be positive definite and $G$ and $H$ Hermitian. We will show that
$$
\frac{d^2}{d\xi^2} \tr\left( 1+ \left( (D+\xi H)^{q/2} K^* (C+\xi G)^p K (D+\xi H)^{q/2} \right)^{-1/(p+q)} \right)^{-1} \leq 0
$$
for all sufficiently small $\xi>0$, which will prove the proposition. To achieve this, we write
$$
\tr\left( 1+ \left( (D+\xi H)^{q/2} K^* (C+\xi G)^p K (D+\xi H)^{q/2} \right)^{-1/(p+q)} \right)^{-1} = \xi \phi(1/\xi)
$$
with
\begin{equation}
\label{eq:defphireal}
\phi(x) := \tr\left( x^{-1} + \left( (x D+H)^{q/2} K^* (xC+G)^p K (xD+H)^{q/2} \right)^{-1/(p+q)} \right)^{-1}
\end{equation}
and appeal to Lemma \ref{epstein}. Thus we need to show that $\phi$ is well-defined on $(R,\infty)$ for some $R>0$ and has an analytic extension to $\C_+$ with non-negative imaginary part.

Let us introduce some notation. For $-\pi\leq\alpha<\beta\leq\pi$ we set
$$
S_{\alpha,\beta}:=\left\{ r e^{i\theta}\in\C :\ r>0\ ,\ \alpha<\theta<\beta \right\}
$$
and (dropping the dimension $N$ from the notation)
$$
\mathcal S_{\alpha,\beta} := \left\{ M\in\mathcal M_N:\  \im \left( e^{-i\alpha} M \right) >0 \,,\ \im\left( e^{-i\beta} M \right)< 0 \right\}.
$$

For $z\in\C_+$ we define
$$
F(z) := (z D + H)^{q/2} K^* (zC+G)^p K (zD+H)^{q/2} \,.
$$
This is well-defined since $\im (zC+G)=(\im z) C >0$ and similarly $\im (zD+H)>0$, and for such matrices the $p$-th and $q/2$-th root are well-defined by analytic functional calculus. Moreover, $F$ is analytic in $\C_+$.

Since $\im(zC+G)>0$ and $\im(zD+H)>0$, one has (see, for instance, \cite[Lemma 1]{Ep})
$$
(zC+G)^p \in\mathcal S_{0,p\pi}
\qquad\text{and}\qquad
(zD+H)^{q/2} \in\mathcal S_{0,q\pi/2} \,,
$$
and, since $K$ is invertible, also $K^*(zC+G)^p K \in\mathcal S_{0,p\pi}$. Therefore (see, for instance, \cite[Lemma 2]{Ep} or, for a simpler proof, \cite[Lemma 10]{AD0})
$$
\spec F(z) \in S_{0,(p+q)\pi} \,.
$$
By the analytic functional calculus we can define $F(z)^{-1/(p+q)}$ for $z\in\C_+$ and by the spectral mapping theorem we obtain
$$
\spec F(z)^{-1/(p+q)} \in \C_- \,.
$$
Therefore $z^{-1} + F(z)^{-1/(p+q)}$ is invertible for $z\in\C_+$ and
$$
\spec\left( z^{-1} + F(z)^{-1/(p+q)} \right)^{-1} \in\C_+ \,.
$$
This proves that
$$
\phi(z) := \tr \left( z^{-1} + F(z)^{-1/(p+q)} \right)^{-1}
$$
is analytic in $\C_+$ and that $\im\phi(z)\geq 0$ for $z\in\C_+$.

Let $x>\max\{ \lambda_{\max}(G)/\lambda_{\min}(C),\lambda_{\max}(H)/\lambda_{\min}(D)\}=:R$, where $\lambda_{\max/\min}(M)$ denote the largest and smallest eigenvalue of a Hermitian matrix $M$. Then $F$ extends continuously from $\C_+$ to $(R,\infty)$ and $F(x)$ is a positive definite matrix for $x>R$. Therefore $\phi$ extends continuously from $\C_+$ to $(R,\infty)$ and the continuation is given by the right side of \eqref{eq:defphireal}. Note that $\phi(x)$ is real (in fact, positive) for $x>R$.
\end{proof}

%%%%%%%%%%%%%%%%%%%%%%%%%

\section{The variational method}\label{sec:var}

In this section we show parts (2) and (3) of Theorem \ref{main}, namely that 
$$
(A,B) \mapsto \tr\left( B^{q/2} K^* A^p K B^{q/2} \right)^s
$$
is jointly convex for $-1\leq q\leq p\leq 0$ and $s>0$ and for $1\leq p\leq 2$, $-1\leq q\leq 0$ and $s\geq\min\{1/(p-1),1/(q+1)\}$. We also prove parts (2) and (3) of Proposition \ref{ups}.

We begin with the proof of Proposition \ref{ups}, since this is simpler and since this will also be needed in the proof of Theorem \ref{main}. We use a variational  method which originates in the work \cite{CL2}. It is based on two ingredients. The first one is a variational characterization of the trace of a power.

\begin{lemma}
If $X\in\mathcal P_N$, then
\begin{equation}
\label{eq:var1}
\tr X^s = \sup_{Y\geq 0} \left( s \tr X Y - (s-1) \tr Y^{s/(s-1)} \right)
\qquad\text{if}\ s> 1 \ \text{or}\ s<0 \,,
\end{equation}
\begin{equation}
\label{eq:var2}
\tr X^s = \inf_{Y>0} \left( s \tr X Y + (1-s) \tr Y^{-s/(1-s)} \right)
\qquad\text{if}\ 0<s<1 \,.
\end{equation}
\end{lemma}

\begin{proof}
We provide two different proofs.

\emph{First proof.} 
We first assume $s>1$. Then by H\"older's and Young's inequalities for any $Y\geq 0$,
\begin{align*}
\tr X Y \leq \left( \tr X^s \right)^{1/s} \left( \tr Y^{s/(s-1)} \right)^{(s-1)/s} \leq \frac 1s \tr X^s + \frac{s-1}{s} \tr Y^{s/(s-1)} \,.
\end{align*}
This proves $\geq$ in \eqref{eq:var1}, and for $\leq$ it suffices to choose $Y=X^s$.

Next, let $0<s<1$. Then by H\"older's and Young's inequalities for any $Y>0$,
\begin{align*}
\tr X^s & = \tr (Y^{s/2} X^s Y^{s/2}) Y^{-s} \leq \left( \tr (Y^{s/2} X^s Y^{s/2})^{1/s} \right)^s \left( \tr Y^{-s/(1-s)} \right)^{1-s} \\
& \leq s \tr (Y^{s/2} X^s Y^{s/2})^{1/s} + (1-s) \tr Y^{-s/(1-s)} \,.
\end{align*}
By the Lieb--Thirring inequality \cite{LT} (see also \cite[Theorem 7.4]{C}),
$$
\tr (Y^{s/2} X^s Y^{s/2})^{1/s} \leq \tr X Y \,,
$$
so we obtain $\geq$ in \eqref{eq:var2}, and for $\leq$ it suffices to choose $Y=X^s$.

Finally, for $s<0$ we apply \eqref{eq:var2} with $s$ replaced by $s/(s-1)\in(0,1)$ and with the roles of $X$ and $Y$ interchanged. We obtain
$$
\tr Y^{s/(s-1)} \leq \frac{s}{s-1} \tr XY + \frac{1}{1-s} \tr X^s \,.
$$
This proves $\geq$ in \eqref{eq:var1}, and for $\leq$ we choose again $Y=X^s$.

\emph{Second proof.}
We provide the details only in the case $0<s<1$ (since in that case before we had to use the Lieb--Thirring inequality). It is easy to see that the infimum on the right side of \eqref{eq:var2} is attained by some $Y_0$ and, differentiating with respect to $Y$, this minimizer satisfies
$$
s X - s Y_0^{-1/(1-s)} =0 \,,
$$
that is, $Y_0 = X^{-1+s}$. Inserting this we obtain
$$
\inf_{Y>0} \left( s \tr X Y + (1-s) \tr Y^{-s/(1-s)} \right)
= s \tr X Y_0 + (1-s) \tr Y_0^{-s/(1-s)} =
\tr X^s \,,
$$
as claimed.
\end{proof}

The second ingredient in this section is a result about suprema and infima of convex functions.

\begin{lemma}\label{convex}
Let $X$ be a convex subset of a vector space, $Y$ a set and $f:X\times Y\to\R$ a function such that $f(\cdot,y)$ is convex for any $y\in Y$.
\begin{enumerate}
\item Then $x\mapsto \sup_{y\in Y} f(x,y)$ is convex.
\item If $Y$ is a convex subset of a vector space and $f$ is jointly convex on $X\times Y$, then $x\mapsto \inf_{y\in Y} f(x,y)$ is convex.
\end{enumerate}
\end{lemma}

\begin{proof}
Let $x_0,x_1\in X$ and $0<\theta<1$. 

(1) Abbreviating $g(x):= \sup_{y\in Y} f(x,y)$, we have for any $y\in Y$,
$$
f((1-\theta)x_0+\theta x_1,y) \leq (1-\theta) f(x_0,y) + \theta f(x_1,y) \leq (1-\theta) g(x_0) + \theta g(x_1) \,.
$$
Taking the supremum over $y$ we obtain $g((1-\theta)x_0+\theta x_1) \leq (1-\theta) g(x_0) + \theta g(x_1)$.

(2) Let $h(x):=\inf_{y\in Y} f(x,y)$. Let $\epsilon>0$ and choose $y_0,y_1\in Y$ such that
$$
f(x_0,y_0)\leq h(x_0) + \epsilon
\qquad\text{and}\qquad
f(x_1,y_1) \leq h(x_1) + \epsilon \,.
$$
Then, by joint convexity,
\begin{align*}
h((1-\theta)x_0+\theta x_1) &
\leq f((1-\theta)x_0+\theta x_1,(1-\theta)y_0+\theta y_1) \\
& \leq (1-\theta) f(x_0,y_0) + \theta f(x_1,y_1) \\
& \leq (1-\theta) h(x_0) + \theta h(x_1) + \epsilon \,.
\end{align*}
Since $\epsilon>0$ is arbitrary, we obtain $h((1-\theta)x_0+\theta x_1) \leq (1-\theta) h(x_0) + \theta h(x_1)$.
\end{proof}

We now use these tools to deduce the convexity assertions in Proposition \ref{ups}.

\begin{proof}[Proof of Proposition \ref{ups}. (2)]
Let $-1\leq p< 0$. We begin with the more difficult case $s< 1$ (which is the only case that will be used in the proof of Theorem \ref{main}). Then, by \eqref{eq:var2},
\begin{align*}
\Upsilon_{p,s}(A)
& = \inf_{Y>0} \left( s \tr K^* A^p K Y + \left(1- s \right) \tr Y^{-s/(1-s)} \right) \\
& = \inf_{C>0} \left( s \tr K^* A^p K C^{1-p} + \left(1- s\right) \tr C^{-s(1-p)/(1-s)} \right) .
\end{align*}
By Ando's convexity theorem \cite{A} (see also \cite[Theorem 6.2]{C}), $(A,C) \mapsto \tr K^* A^p K C^{1-p}$ is jointly convex. Moreover, since $-s(1-p)/(1-s)<0$, $C \mapsto \tr C^{-\frac {s(1-p)}{1-s}}$ is convex. Thus, by part (2) of Lemma~\ref{convex}, $\Upsilon_{p,s}$ is convex.

We now consider the case $s>1$ and therefore
\begin{align*}
\Upsilon_{p,s}(A)
& = \sup_{Y\geq 0} \left( s \tr K^* A^p K Y - \left( s - 1\right) \tr Y^{s/(s-1)} \right).
\end{align*}
Since $A\mapsto A^p$ is operator convex, $A \mapsto \tr K^* A^p K Y$ is convex. Thus, by part (1) of Lemma \ref{convex}, $\Upsilon_{p,s}$ is convex.

The case $s=1$ is even simpler and follows directly from the operator convexity of $A\mapsto A^p$.
\end{proof}

\begin{proof}[Proof of Proposition \ref{ups}. (3)]
The argument is similar to that in part (1) and we refer to \cite{CL2} for details . This result will not be needed in the proof of Theorem~\ref{main}.
\end{proof}

Before we turn to the proof of Theorem \ref{main} we recall a joint operator convexity statement from \cite{CFL}.

\begin{lemma}\label{opconv}
For any $-1\leq q\leq 0$ and any $K\in\mathcal M_N$ the map
$$
\mathcal P_N\times\mathcal P_N \ni (A,B)\mapsto AKB^qK^* A
$$
is jointly convex.
\end{lemma}

In fact, in \cite[Theorem 3.2]{CFL} we also proved that the restriction $-1\leq q\leq 0$ is necessary and that convexity does not hold if $A$ is raised to a non-zero power.

In \cite{CFL} we observed that this lemma follows from \cite[Corollary 2.1]{L}. (The presence of the operators $K$ and $K^*$ does not present a problem. They can be dealt with by doubling of dimension as, for instance, in \cite[Lemma 1.1]{CFL}. In this way they can be made into unitary operators and then they can be absorbed into $B$.) Also, in \cite[Remark 3.5]{CFL} we explained a simple alternative proof which reduces the case $-1<q<0$ to the well-known case $q=-1$ \cite{K,LR}.

We now prove the convexity assertions in Theorem \ref{main}.

\begin{proof}[Proof of Theorem \ref{main}. (2)]
We assume that $-1\leq q\leq p \leq 0$ and that $s>0$. By an approximation argument we may assume that $K$ is invertible and we denote $L := K^{-*}$. Then, by \eqref{eq:var1},
\begin{align*}
\Psi_{p,q,s}(A,B)
& = \tr \left( B^{-q/2} L^* A^{-p} L B^{-q/2} \right)^{-s} \\
& = \sup_{Y\geq 0} \left( -s \tr B^{-q/2} L^* A^{-p} L B^{-q/2} Y + (s+1) \tr Y^\frac{s}{s+1} \right) \\
& = \sup_{Z \geq 0} \left( -s \tr L^* A^{-p} L Z + (s+1) \tr (B^{q/2} Z B^{q/2})^\frac{s}{s+1} \right).
\end{align*}
Since $A\mapsto A^{-p}$ is operator concave, $A\mapsto -s \tr L^* A^{-p} L Z$ is convex for any $Z$. Moreover, by part (1) of Proposition \ref{ups}, $B\mapsto \tr (Z^{1/2} B^{q} Z^{1/2})^\frac{s}{s+1}=\tr(B^{q/2}ZB^{q/2})^\frac{s}{s+1}$ is convex for any $Z$. (We apply the lemma with $A$ replaced by $B$, $K$ by $Z^{1/2}$, $p$ by $q$ and $s$ by $s/(s+1)$.) Thus, by part (1) of Lemma \ref{convex}, $\Psi_{p,q,s}$ is convex.
\end{proof}

\begin{proof}[Proof of Theorem \ref{main}. (3)]
We break the proof into three steps.

\emph{Step 1.} We assume $1\leq p\leq 2$, $-1\leq q\leq 0$ and $s\geq 1/(1+q)$. Then, by \eqref{eq:var1},
\begin{align*}
\Psi_{p,q,s}(A,B) & = \sup_{Y\geq 0} \left( s \tr B^{q/2} K^* A^p K B^{q/2} Y - (s-1) \tr Y^{s/(s-1)} \right) \\
& = \sup_{Z\geq 0} \left( s \tr K^* A^p K Z - (s-1) \tr \left( B^{-q/2} Z B^{-q/2} \right)^{s/(s-1)} \right).
\end{align*}
Since $A\mapsto A^p$ is operator convex, $A\mapsto s \tr K^* A^p K Z$ is convex for any $Z\geq 0$. Moreover, by part (1) of Proposition \ref{ups} and since $\frac{s}{s-1}\leq -\frac1q$, $B\mapsto \tr \left( Z^{1/2} B^{-q} Z^{1/2} \right)^{\frac s{s-1}} = \tr\left( B^{-q/2} Z B^{-q/2} \right)^{\frac s{s-1}}$ is concave. Thus, by part (1) of Lemma \ref{convex}, $\Psi_{p,q,s}$ is convex.

\emph{Step 2.} We assume $1< p\leq 2$, $-1\leq q\leq 0$ and $s\geq 1/(p-1)$. First, assume that $s=1$, so that necessarily $p=2$. The joint convexity follows from \eqref{eq:commute} and the fact that according to Lemma \ref{opconv} $(A,B)\mapsto AK B^q K^* A$ is operator convex.

Now let $s>1$ (and still $s\geq 1/(p-1)$). Then, by \eqref{eq:commute} and \eqref{eq:var1},
\begin{align*}
\Psi_{p,q,s}(A,B) & = \sup_{Y> 0} \left( s \tr A^{p/2} K B^q K^* A^{p/2} Y - (s-1) \tr Y^{s/(s-1)} \right) \\
& = \sup_{Z> 0} \left( s \tr A K B^q K^* A Z - (s-1) \tr \left( A^{1-p/2} Z A^{1-p/2} \right)^{s/(s-1)} \right).
\end{align*}
Again by Lemma \ref{opconv}, $(A,B)\mapsto \tr A K B^q K^* A Z$ is jointly convex. Moreover, by part (1) of Proposition \ref{ups} and since $s/(s-1)\leq 1/(2-p)$, $A\mapsto \tr \left( Z^{1/2} A^{2-p} Z^{1/2} \right)^{s/(s-1)} = \tr \left( A^{1-p/2} Z A^{1-p/2} \right)^{s/(s-1)}$ is concave. Thus, by part (1) of Lemma \ref{convex}, $\Psi_{p,q,s}$ is convex.

\emph{Step 3.} We assume that $p=2$ and $1/(2+q)\leq s<1$. Then, by \eqref{eq:commute} and \eqref{eq:var2},
\begin{align*}
\Psi_{2,q,s}(A,B) & = \inf_{Y> 0} \left( s \tr A K B^q K^* A Y + (1-s) \tr Y^{-s/(1-s)} \right) \\
& = \inf_{Z> 0} \left( s \tr A K B^q K^* A Z^{1-1/s} + (1-s) \tr Z \right).
\end{align*}
By \cite[Corollary 2.1]{L} and since $q+1-1/s\geq -1$, $(A,B,Z)\mapsto \tr A K B^q K^* A Z^{1-1/s}$ is jointly convex. (Note that the operator $K$, which is not present in \cite[Corollary 2.1]{L}, can be dealt with as explained after Lemma \ref{opconv}.) Thus, by part (2) of Lemma \ref{convex}, $\Psi_{2,q,s}$ is convex.
\end{proof}

\begin{remark}
Both in Steps 1 and 2 of the previous proof we applied part (1) of Proposition \ref{ups}. We have proved the latter, as a special case of part (1) of Theorem \ref{main}, using the complex analysis method. It is interesting to note, however, that in Steps 1 and 2 of the previous proof we applied part (1) of Proposition \ref{ups} only with a power $s\geq 1$. For such powers part (1) of Proposition \ref{ups} can be proved also using the variational method; see \cite[Theorem 1.1]{CL2}.
\end{remark}

\begin{remark}
In the special case $p=2$, $-1\leq q\leq 0$, $s\geq 1$ there is a proof which is only based on Lemma \ref{opconv} and which does not use Proposition \ref{ups}; see \cite[Remark 4.3]{CFL}.
\end{remark}

\begin{remark}
Let us give an alternative proof for $p=2$ and $1/(2+q)\leq s<1$, which does not use the deep \cite[Corollary 2.1]{L}, but only the special case $q=-1$ of Lemma \ref{opconv}, plus part (3) of Theorem \ref{main} for $p=1$. (Recall that Lemma \ref{opconv} can be deduced from \cite[Corollary 2.1]{L}, but that its special case $q=-1$ is rather simple and well-known.) We do \emph{not} use \eqref{eq:commute}, but use directly \eqref{eq:var2} to write
\begin{align*}
\Psi_{2,q,s}(A,B) & = \inf_{Y> 0} \left( s \tr B^{q/2} K^* A^2 K B^{q/2} Y + (1-s) \tr Y^{-s/(1-s)} \right) \\
& = \inf_{Z> 0} \left( s \tr K^* A^2 K Z^{-1} + (1-s) \tr (B^{-q/2} Z^{-1} B^{-q/2})^{-s/(1-s)} \right).
\end{align*}
By Lemma \ref{opconv}, $(A,Z)\mapsto \tr K^* A^2 K Z^{-1}$ is jointly convex. Moreover, by Step 1 in the proof of part (3) of Theorem \ref{main}, $(B,Z)\mapsto \tr(Z^{1/2} B^q Z^{1/2})^{\frac s{1-s}} = \tr(B^{-q/2} Z^{-1} B^{-q/2})^{-\frac s{1-s}}$ is jointly convex. Thus, by part (2) of Lemma \ref{convex}, $\Psi_{2,q,s}$ is convex.
\end{remark}

%%%%%%%%%%%%%%%%%%%%%

\section{Necessary conditions}

In this section we reproduce the arguments from \cite{B,CL2,Hi2} to prove Propositions \ref{upsnec} and \ref{mainnec} containing necessary conditions for concavity and convexity.

\begin{proof}[Proof of Proposition \ref{upsnec}]
(1) Clearly, concavity of $\Upsilon_{p,s}$ on $\mathcal P_1$ implies that $(0,\infty)\ni a\mapsto a^{ps}$ is concave and therefore $0<ps\leq 1$. Moreover, taking
$$
A = \begin{pmatrix}
a & 0 \\ 0 & b
\end{pmatrix}
\qquad\text{and}\qquad
K = \begin{pmatrix} 1 & 0 \\ 1 & \epsilon \end{pmatrix}
$$
with numbers $a,b$ and $\epsilon$ and letting $\epsilon\to 0$, we deduce from concavity on $\mathcal P_2$ that $(0,\infty)\times(0,\infty)\ni (a,b)\mapsto (a^p+b^p)^s$ is jointly concave. By differentiating twice with respect to $a$ and evaluating at $a\ll b$ we find that $p\leq 1$.

(2) Clearly, convexity of $\Upsilon_{p,s}$ on $\mathcal P_1$ implies that $(0,\infty)\ni a\mapsto a^{ps}$ is convex and therefore $0<ps\leq 1$ or $ps<0$. If $s=1$, then the convexity of $\Upsilon_{p,s}$ on $\mathcal P_2$ for any (invertible) $K$ implies that $\mathcal P_2\ni A\mapsto A^p$ is operator convex. As is well known, this implies that $-1\leq p<0$ or $1\leq p\leq 2$. Now let $s\neq 1$. Taking
$$
A = \begin{pmatrix}
a & 0 \\ 0 & b
\end{pmatrix}
\qquad\text{and}\qquad
K = \begin{pmatrix} 1 & 0 \\ 1 & \epsilon \end{pmatrix}
$$
where now $a,b$ and $1$ are $2\times 2$ matrices and letting $\epsilon\to 0$, we deduce from convexity on $\mathcal P_4$ that $\mathcal P_2 \times\mathcal P_2 \ni (A,B)\mapsto \tr (A^p+B^p)^s$ is jointly convex. In particular, for any $t>0$, $\mathcal P_2\ni A\mapsto \tr (t A^p + B^p)^s$ is convex. Using
$$
\tr (t A^p + B^p)^s = \tr B^{ps} + st \tr B^{p(s-1)} A^p + o(t)
\qquad\text{as}\ t\to 0
$$
we deduce that $\mathcal P_2\ni A\mapsto\tr B^{p(s-1)} A^p$ is convex, which (since $s\neq 1$) means that $\mathcal P_2\ni A\mapsto A^p$ is operator convex. As before, this implies $-1\leq p<0$ or $1\leq p\leq 2$.
\end{proof}

\begin{proof}[Proof of Proposition \ref{mainnec}]
(1) Clearly, concavity of $\Psi_{p,q,s}$ on $\mathcal P_1\times\mathcal P_1$ implies that $(0,\infty)\ni a\mapsto a^{(p+q)s}$ is concave and therefore $0\le (p+q)s\leq 1$. Writing an invertible $K$ as $K=U|K|$ with $U$ unitary, chosing $B^{q/2} = |K|$ and absorbing $U$ into $A$ we deduce from the concavity of $\mathcal P_2\ni A\mapsto\Psi_{p,q,s}(A,B)$ that $\mathcal P_2 \ni A \mapsto\Upsilon_{p,s}(A)$ is concave. According to part (1) of Proposition \ref{upsnec} this implies $0\leq p\leq 1$. Exchanging the roles of $p$ and $q$ we find that $0\leq q\leq 1$.

(2) Clearly, convexity of $\Psi_{p,q,s}$ on $\mathcal P_1\times\mathcal P_1$ implies that $(0,\infty)\times(0,\infty)\ni (a,b)\mapsto a^{ps}b^{qs}$ is jointly concave. By an elementary analysis of the Hessian we conclude that, if $p\geq 0$, then $q\leq 0$ and $(p+q)s\geq 1$. Similarly as in the first part of the proof, part (2) of Proposition \ref{upsnec} implies either $-1\leq p\leq 0$ or $1\leq p\leq 2$ and, exchanging the roles of $p$ and $q$, either $-1\leq q\leq 0$ or $1\leq q\leq 2$. This corresponds to four disjoint squares in the $(p,q)$ plane. The square $-1\leq p\leq 0$, $1\leq q\leq 2$ is excluded by our assumption $p\geq q$ and the square $1\leq p,q\leq 2$ is excluded by the above elementary analysis. This concludes the proof.
\end{proof}

%%%%%%%%%%%%%%%%%%%%%%%%%%%%%%%%%%%%%%%%%%%%%%%%%%%%%%%%%%%%%%%%%%%%%%%%%%%%%%%%
%%%%%%%%%%%

\bibliographystyle{amsalpha}

\end{document}